\documentclass[
 reprint,
 amsmath,amssymb,
 aps,
]{revtex4-1}

\usepackage{amsthm}
\usepackage{graphicx}
\usepackage{dcolumn}
\usepackage{bm}
\usepackage{hyperref}
\usepackage{mathrsfs}

\newtheorem{Theorem}{Theorem}

\begin{document}

\preprint{APS/123-QED}

\title{Informational power of the Hoggar SIC-POVM}

\author{Anna Szymusiak}
\altaffiliation{anna.szymusiak@uj.edu.pl}
\author{Wojciech S{\l}omczy\'{n}ski}
\altaffiliation{wojciech.slomczynski@im.uj.edu.pl}
\affiliation{Institute of Mathematics, Jagiellonian University, \L ojasiewicza 6, 30-348 Krak\'{o}w, Poland}

%\date{\today}

\begin{abstract}
We compute the informational power for the Hoggar SIC-POVM in dimension $8$, i.e.\ the classical capacity
of a quantum-classical channel generated by this measurement. We show that the states constituting a maximally
informative ensemble form a twin Hoggar SIC-POVM being the image of the original one under a conjugation.
\begin{description}

\item[Mathematics Subject Classification]
81P15, 94A17, 94A40, 58D19
\item[PACS numbers]
03.65.Ta, 03.67.-a
%\item[Keywords]
%SIC-POVM, informational power, entropy, Hoggar lines

\end{description}
\end{abstract}

\maketitle
\section{Introduction}
Among positive operator valued measures (POVMs) representing general quantum measurements,
symmetric informationally complete (SIC) POVMs, called by Christopher Fuchs `mysterious
entities', play a special role. On the one hand, they are crucial ingredients of the Quantum
Bayesianism (or QBism) approach to the foundations of quantum physics proposed fifteen years ago by
Caves, Fuchs and Schack \cite{CavFucSch02,AppEriFuc11}, on the other hand, they are widely used in
various areas of quantum information theory like quantum cryptography \cite{Ren05}, quantum state
tomography \cite{Sco06,ZhuEng11,Jia13,Benetal15}, quantum communication \cite{Oreetal11} or entanglement
detection \cite{XiZhe15}, see also \cite{Zhu12,AppDanFuc14}.

However, despite many efforts as well as positive results obtained for lower dimensions, see e.g.\ \cite{Zhu12},
these important objects remain elusive, as the problem of their existence in arbitrary dimension is still
open. Recently, this question has been reformulated in the language of various algebraic structures
(Lie groups, Lie algebras, and Jordan algebras) \cite{AppFucZhu15}, but it has also simple
interpretation in terms of metric spaces. Namely, the existence of SIC-POVMs in dimension $d$ is
equivalent to the fact that the equilateral dimension (i.e.\ the maximum number of equidistant
points) \cite{Man14,HugSal16} of $d$-dimensional complex projective space endowed with the Fubini-Study metric equals $d^2$.

The eight-dimensional Hoggar lines \cite{Hog81} provide one of the first examples of SIC-POVMs
found in dimension larger than two. It seems that this set exhibits a higher level of symmetry
than most known SIC-POVMs, and, at the same time, its symmetry has a slightly different character than
in case of all other known SIC-POVMs. This, using Blakean language, `fearful symmetry' of Hoggar lines
makes it especially interesting object of study.

The informational power of a quantum measurement, that is the maximum amount of classical information
that it can extract from any ensemble of quantum states \cite{DAretal14a}, being equal to the classical
capacity of a quantum-classical channel generated by this measurement, has received much attention in recent
years \cite{DAretal11,Oreetal11,Hol12,DAretal14b,DAr14,Szy14a,DAr15,SloSzy16}. However, this
quantity is in general not easy to compute analytically, especially in higher dimensions.
In this paper we show that the informational power of the Hoggar SIC-POVM is equal to $2\ln(4/3)$.
To this aim we use the construction of Hoggar lines newly discovered by Jedwab and Wiebe \cite{JedWie15}.
As a corollary we get that the bound for the informational power of 2-designs (including SIC-POVMs)
obtained recently by Dall'Arno \cite{DAr15} is saturated in dimension eight. Moreover, we show that
a maximally informative ensemble for a Hoggar SIC-POVM forms another `twin' Hoggar SIC-POVM being
the image of the original one under a (complex) conjugation, i.e.\ an antiunitary involutive map,
and sharing the same symmetries as the original one.

\section{SIC-POVMs}

With any finite-dimensional quantum system one can associate a complex Hilbert space $\mathbb C^d$.
Then the pure states $\mathcal P(\mathbb C^d)$ of the system are described by one-dimensional
orthogonal projections, that is
$\mathcal P(\mathbb C^d):=\{\rho\in\mathcal L(\mathbb C^d)|\rho\geq 0, \rho^2=\rho, \operatorname{tr}(\rho)=1\}$,
and the mixed states $\mathcal S(\mathbb C^d)$ are convex combinations of pure states, i.e.\
density operators on $\mathbb C^d$.

A general quantum measurement is described by a \emph{positive operator valued measure} (\emph{POVM}).
In this paper we consider the discrete version of it, i.e.\ by POVM we mean a set $\Pi:=\{\Pi_j\}_{j=1}^k$
of nonzero positive semi-definite operators on $\mathbb C^d$ satisfying the identity decomposition:
$\sum_{j=1}^k\Pi_j=\textrm I_d$. In this framework the probability of obtaining $j$-th ($j=1,\dots,k$)
outcome, given that the initial (pre-measurement) state of the system was $\rho \in \mathcal S(\mathbb C^d)$,
is equal to $p_j(\rho,\Pi):=\operatorname{tr}(\rho\Pi_j)$.

Among quantum measurements we can distinguish \emph{symmetric informationally complete} (\emph{SIC}) \emph{POVMs}, i.e.\
POVMs consisting of $d^2$ subnormalized rank-one projectors $\Pi_j:=|\phi_j\rangle\langle\phi_j|/d$ ($j=1,\dots,k$)
with equal pairwise Hilbert-Schmidt inner products: $\operatorname{tr}(\Pi_i\Pi_j)=|\langle\phi_i|\phi_j\rangle|^2/d^2=1/(d^2(d+1))$
for $i\neq j, i,j=1,\dots,k$, where $\phi_j$ are elements of the unit sphere in $\mathbb C^d$ determined up
to a phase factor. Note that this condition implies that SIC-POVMs are indeed \emph{informationally complete} (\emph{IC}),
i.e.\ the statistics of measurement uniquely determine the initial state \cite{Renetal04}. Since any IC-POVM must contain
at least $d^2$ elements, SIC-POVMs are special examples of \emph{minimal} IC-POVMs. If the pre-measurement pure state
is given by $|\psi\rangle\langle\psi|$, where $\psi$ is an element of the unit sphere in $\mathbb C^d$, then
$p_j(|\psi\rangle\langle\psi|,\Pi)= |\langle\psi|\phi_j\rangle|^2/d$.

Furthermore, let us recall that a \emph{complex projective \linebreak $t$-design} ($t\in\mathbb N$) is a set $\{\rho_j\}_{j=1}^k$ of pure states such that
\begin{eqnarray*}
\frac{1}{k^2}\sum_{j,m=1}^k f(\operatorname{tr}(\rho_j\rho_m))=
\iint\limits_{\mathcal P(\mathbb C^d)\times \mathcal P(\mathbb C^d)} f(\operatorname{tr}(\rho\sigma))\textrm{d}\mu(\rho)\textrm{d}\mu(\sigma)
\end{eqnarray*}
for every real-valued polynomial $f$ of degree $t$ or less, where $\mu$ stands for the unique unitarily invariant (Fubini-Study) probabilistic
measure on $\mathcal P(\mathbb C^d)$ \cite{Sco06}. The SIC-POVMs can be equivalently described as complex projective 2-designs (called also
\emph{spherical quantum} 2-\emph{designs}) with $d^2$ elements, see e.g.\ \cite{Renetal04}.

\section{Informational power}
\label{Infpow}

The indeterminacy of quantum measurement $\Pi:=\{\Pi_j\}_{j=1}^k$ can be quantized by a number that
characterizes the randomness of the distribution of measurements outcomes $(p_j(\rho,\Pi))_{j=1}^k$
depending on the pre-measurement state of the system $\rho \in \mathcal S(\mathbb C^d)$.
The most natural choice for such a tool is the \emph{Shannon entropy}. Thus, by the \emph{entropy of the measurement}
$\Pi$ we mean a function $H(\cdot,\Pi):\mathcal S(\mathbb C^d)\to \mathbb R$ defined by
$$H(\rho,\Pi):=\sum_{j=1}^k\eta(p_j(\rho,\Pi)),$$
where the \emph{Shannon entropy function} $\eta:[0,1]\to\mathbb R$ is given by $\eta(t):=-t\ln t$ for $t>0$
and $\eta(0):=0$; see \cite{SloSzy16} for the history and interpretation of this notion.
It follows from the concavity of $H$ that this function attains minima in the set of pure states, finding
the minimizers, however, is not a trivial task in general, even for SIC-POVMs, where only the results
for dimension two \cite{Oreetal11,SloSzy16} and three \cite{Szy14a}
has been known. In fact, the latter result was proven under the assumption that a SIC-POVM is covariant, but it follows
from \cite{HugSal16} that all SIC-POVMs in dimension three share this property. On the other hand, for an arbitrary
SIC-POVM, the maximum value of $H$ for pure pre-measurement states is equal to $((d-1)/d)\ln(d+1)$ \cite{Szy16}.

Let us now consider an ensemble $\mathcal{E}=\left(\tau_{i},p_{i}\right)_{i=1}^{m}$, where $p_{i}\geq0$ are
\emph{a priori} probabilities of density matrices \linebreak $\tau_{i}\in\mathcal{S}%
\left(  \mathbb{C}^{d}\right)  $, for $i=1,\ldots,m$, and $\sum
\nolimits_{i=1}^{m}p_{i}=1$. The \textsl{mutual information} between
$\mathcal{E}$ and $\Pi$ is given by:

$$I(\mathcal{E},\Pi):=\sum_{i=1}^m\eta(\sum_{j=1}^k P_{ij})+\sum_{j=1}^k\eta(\sum_{i=1}^m P_{ij})-\sum_{i=1}^m\sum_{j=1}^k\eta(P_{ij}),$$
where $P_{ij} := p_i\operatorname{tr}(\tau_i\Pi_j)$ for $i=1,\ldots, m$ and $j=1,\ldots,k$. This quantity can be considered
as a measure of how much information can be extracted from ensemble $\mathcal{E}$ by measurement $\Pi$. Thus
the following two questions arise: what is the maximum amount of information one can get from the
given ensemble (i.e.\ $\max_\Pi I(\mathcal{E},\Pi)$, studied, e.g.\ in \cite{Hol73,Dav78,Sasetal99}) and what
is the capability of extracting information by given measurement (i.e.\ $\max_\mathcal{E} I(\mathcal{E},\Pi)$, examined
in \cite{DAretal11,Oreetal11,DAretal14a,DAretal14b,SloSzy16}). The latter quantity, denoted
by $W(\Pi)$, is called the \emph{informational power} of $\Pi$.

Both the minimum entropy and informational power of $\Pi$ can
also be interpreted in terms of the \emph{quantum-classical channel}
$\Phi:\mathcal S(\mathbb C^d)\to\mathcal S(\mathbb C^k)$ generated
by $\Pi$ and given by $\Phi(\rho):=\sum_{j=1}^k\operatorname{tr}(\rho\Pi_j)|e_j\rangle\langle e_j|$
for some orthonormal basis $(|e_j\rangle)_{j=1}^k$ in $\mathbb C^k$.
The former quantity is equal to the \emph{minimum output entropy} of $\Phi$, $\min_\rho S(\Phi(\rho))$,
where $S$ is the \emph{von Neumann entropy} defined by $S(\tau): = -\operatorname{tr}(\tau\ln\tau)$
for $\tau \in \mathcal S(\mathbb C^d)$ \cite{Sho02}.
The latter one is just the \emph{classical capacity} $\chi(\Phi)$ of the channel $\Phi$,
given by
$\chi(\Phi):=\max_{\mathcal{E}=\left(\tau_{i},p_{i}\right)_{i=1}^{m}}\left\{S\left(\sum_{i=1}^m p_i\Phi(\tau_i)\right)-\sum_{i=1}^m p_i S(\Phi(\tau_i))\right\}$
\cite{Hol12,Oreetal11}.

The minimal entropy of $\Pi$ and its informational power are related by:
\begin{equation}
\label{infpowent}
W(\Pi)\leq\ln k-\min_{\rho \in \mathcal S(\mathbb C^d)} H(\rho,\Pi)
\end{equation}
and the equality holds if and only if there exists an ensemble $\mathcal{E}=(p_i,\tau_i)_{i=1}^m$
such that the states $\tau_i$ ($i=1,\ldots,m$) are minimizers of $H(\cdot,\Pi)$ and
$\operatorname{tr}((\sum_{i=1}^m p_i\tau_i)\Pi_j)=1/k$ for  $j=1,\ldots,k$ \cite[Proposition 6]{SloSzy16}.
This condition is in particular fulfilled if we assume that $\Pi$ is covariant with respect to an irreducible
representation, a fact already observed by Holevo \cite{Hol05}. To see this,
it is enough to consider the ensemble consisting of equiprobable elements of the orbit of any minimizer
of H under the action of this representation.

So far the informational power has been computed analytically in few cases only:
for all highly symmetric POVMs in dimension two: seven sporadic measurements, including the `tetrahedral' SIC-POVM,
and one infinite series \cite{SloSzy16} (though for some of them the result was known earlier,
see \cite{Sch89,San95,Ghietal03,DAretal11,Oreetal11}), the SIC-POVMs in dimension three \cite{Szy14a},
and the POVMs consisting of four MUBs, again in dimension three \cite{DAr14}.
The first two results has been obtained with the method developed in \cite{SloSzy16} based on the Hermite interpolation
of Shannon entropy function. In this paper we enlarge this collection, computing the informational power for the Hoggar SIC-POVM.

Let us recall that for SIC-POVMs in dimension $d$ the sum of squared probabilities of the measurement outcomes (so called
\emph{index of coincidence}, known under various names in the literature, see \cite[Sec.~8]{Ell13})
is the same for each initial pure state and equals to $r:=2/(d(d+1))$. The problem of finding the minimum of the Shannon entropy
under assumption that the index of coincidence is equal to $r$ was analyzed by Harremo\"{e}s and Tops\o{}e in \cite[Theorem~2.5]{HarTop01},
see also \cite{Zyc03}. From their result one can deduce that if $1/r \in \mathbb{N}$, then this minimum is attained
for the probability distribution  $(r,\ldots,r,0,\ldots,0)$ with $1/r$
probabilities equal to $r$, and the rest equal to $0$. Hence, the minimum entropy of a SIC-POVM
is bounded from below by $\ln(d(d+1)/2)$, and using inequality (\ref{infpowent}) with $k=d^2$, we get that its informational power
is bounded from above by $\ln(2d/(d+1))$, see also \cite[Corollary 2]{DAr15}. The achievability of this bound in dimension $d$
is equivalent to the existence of a vector (representing pure state) orthogonal to $(d-1)d/2$ elements of a SIC-POVM, and making equal
angles with $d(d+1)/2$ others, the problem already analyzed in \cite{AppEriFuc11}. Consequently, this bound is achieved for SIC-POVMs in dimensions 2 and 3,
but numerical results suggest that this is not the case for known SIC-POVMs in dimensions 4 and 5 \cite{AppEriFuc11,DAr15}. We shall see that this
bound is achieved again in dimension 8 for the Hoggar SIC-POVM.

\section{Hoggar lines and their symmetries}

The \textit{Hoggar (lines)} SIC-POVM ($HL$) was constructed with the help of
computer by Hoggar in \cite{Hog81} as the complexification of $64$ lines
through the origin in the four-dimensional quaternionic space, or more
precisely, as the set of diameters of a quaternionic polytope with $128$
vertices. In fact, he had announced this result as early as in \cite{Hog78},
and in \cite{Hog98} gave a computer independent proof that these lines are
equiangular. One year later, Zauner showed in his thesis \cite{Zau99} that
this SIC-POVM is covariant with respect to $P_3$, the quotient
of the \textit{three-qubit Pauli group} (called also the \textit{Galoisian
variant of the discrete Weyl-Heisenberg group in dimension $8$})
by its center, which is group-theoretically isomorphic to $\left(
\mathbb{Z}_{2}\otimes\mathbb{Z}_{2}\right)^{\otimes 3}$. Quite recently, Zhu
\cite[Sec.8.6]{Zhu12} proved the long expected result that the Hoggar lines are not
projectively equivalent to any SIC-POVM covariant with respect to the group
$\mathbb{Z}_{8} \otimes\mathbb{Z}_{8}$ isomorphic to the quotient
of the usual discrete Weyl-Heisenberg group in dimension $8$ by its center. 
As the Hoggar SIC-POVM is currently the only
known such example in any dimension, this property makes this object
exceptional among SIC-POVMs. In the present paper by a Hoggar SIC-POVM
we mean any SIC-POVM\ projectively equivalent to the original Hoggar construction.

In his thesis \cite[Sec.10.4]{Zhu12} Zhu analyzed
the extended symmetry group of the Hoggar lines, $\operatorname*{Sym}\left(HL\right)$,
i.e.\ the subgroup of the projective unitary-antiunitary group
$\operatorname*{PUA}\left(  \mathbb{C}^{8}\right)$ leaving this set
invariant, and showed that it has $774\,144$ elements. Zhu proved also that
$\operatorname*{Sym}\left(HL\right)$ is a subgroup of the \textit{extended
multiqubit Clifford collineation group} ($\overline{\mathcal{EC}}(8)$)
of the three-qubit Pauli group, i.e.\ its normalizer within
$\operatorname*{PUA}\left(  \mathbb{C}^{8}\right)  $, having $240\times
774\,144$ elements. Analogously, the unitary symmetry group of the Hoggar lines $\operatorname*{Sym}_{U}\left(HL\right)$
is a subgroup of order $387\,072$ of the
\emph{multiqubit Clifford collineation group} $\overline{\mathcal{C}}(8)$ with
$240 \times 387\,072$ elements. Thus, the orbit of any state from
$HL$ under the action of the (extended) Clifford group is the union of $240$
copies of $HL$. It was proved recently in \cite{Web15,Zhu15c}
that this set constitutes a $3$-design in $\mathcal P(\mathbb C^8)$.

\newpage

It is well known that $\overline{\mathcal{C}}(8)$ is a (unique) non-split extension of the symplectic
group $Sp(6,2)$ by $P_3$ \cite{Boltetal61a,Boltetal61b,Dem74,BasMoo13}.
It means that $\overline{\mathcal{C}}(8)$ acts on $P_3$ by conjugation as $Sp(6,2)$, but
$Sp(6,2)$ is not embeddable in $\overline{\mathcal{C}}(8)$.
In yet another words, though the elements of $\overline{\mathcal{C}}(8)$
can be labelled by the elements of the set $P_3 \times Sp(6,2)$, this group is not a semidirect product
of $P_3$ by $Sp(6,2)$, and, in particular, the product of two elements from $\overline{\mathcal{C}}(8)$ labelled by $(0,M_1)$ and $(0,M_2)$ for $M_1,M_2 \in Sp(6,2)$ may have non-zero first coordinate, see
\cite[Thm. 2]{DehMoo03}.

Let $\psi$ be a \emph{fiducial} vector for $HL$, i.e.\ one of the vectors from $\mathcal P(\mathbb C^8)$
generating $HL=(P_3)\psi$. Then, it is easy to show that $\operatorname*{Sym}_{U}\left(HL\right)
= P_3 \rtimes \left(\operatorname*{Sym}_{U}\left(HL\right)\right)_\psi$, where
$\left(\operatorname*{Sym}_{U}\left(HL\right)\right)_\psi$ is the stabilizer of
$\psi$ in $\operatorname*{Sym}_{U}\left(HL\right)$.
In consequence, $\left(\operatorname*{Sym}_{U}\left(HL\right)\right)_\psi \simeq
\operatorname*{Sym}_{U}\left(HL\right) / P_3$ is a subgroup of
$\overline{\mathcal{C}}(8) / P_3 \simeq Sp(6,2)$.
Moreover, we know from \cite[Sec.10.4]{Zhu12} that
$\left(\operatorname*{Sym}_{U}\left(HL\right)\right)_\psi$ has $6\,048$ elements.
However, there is only one (up to isomorphism) subgroup of $Sp(6,2)$ of order $6\,048$,
namely the derived Chevalley group $G_2'(2)$ \cite{ConLee13}. Consequently,
$\left(\operatorname*{Sym}_{U}\left(HL\right)\right)_\psi \simeq G_2'(2)$, and
so $\operatorname*{Sym}_{U}\left(HL\right) \simeq P_3 \rtimes G_2'(2)$.
However, in spite of the fact that
$(\operatorname*{Sym}_{U}\left(HL\right))_\psi \rtimes \mathbb{Z}_2\simeq
(\operatorname*{Sym}\left(HL\right))_\psi $ and
$G_2'(2) \rtimes \mathbb{Z}_2 \simeq G_2(2)$,
where $G_2(2)$ is the Chevalley group of order $12096$, it is not clear whether
$(\operatorname*{Sym}\left(HL\right))_\psi \simeq G_2(2)$.

It is natural to consider normalized rank-$1$ POVMs as subsets of the complex
projective space. It seems that the Hoggar lines are exceptional also in this
context. Clearly, they form a symmetric set, as every SIC-POVM known so far does,
but in fact they exhibit higher level of symmetry. Together with the `tetrahedral'
SIC-POVM in dimension two and the Hesse SIC-POVM in dimension three,
they are the only SIC-POVMs that are \textit{super-symmetric}, which means that
$\operatorname*{Sym}\left(  HL\right)  $ acts doubly-transitively on $HL$
\cite[Theorem~1]{Zhu15a}. As a consequence, one can deduce \cite[Corollary~1]{SloSzy16}
that they form a \textit{highly symmetric} subset of $\mathbb{C}P^{7}$ in the sense of
\cite{SloSzy16}, see also \cite{Zhu15a}.

There exist other constructions of $HL$ that were proposed by Grassl
\cite[Sec.~4.2.2]{Gra05} (the fact that his construction does indeed lead to the
set of Hoggar lines was observed later by Zhu \cite{Zhu12}), Godsil \& Roy \cite{GodRoy09},
Jiangwei \cite{Jia13}, and Jedwab \& Wiebe
\cite{Wie13,JedWie15,JedWie16}. We shall use the last of these
in the present paper.

\section{Main results}
Let us recall that a \emph{complex Hadamard matrix} $H=(h_{ij})_{i,j=0}^{d-1}$ is a $d\times d$ matrix such that $|h_{ij}|^2=1$ for $i,j=0,\ldots,d-1$, and
$$HH^\dagger=d\thinspace \textrm{I}_d.$$
In particular, if all its entries lie in $\{-1,1\}$, then $H$ is called a \emph{real Hadamard matrix}. In this case
\begin{equation}\label{diag}
\sum_{l=0}^{d-1} h_{jl}^2=d,\quad \textrm{for }j=0,\dots,d-1
\end{equation}
and
\begin{equation}\label{offdiag}
\sum_{l=0}^{d-1}h_{jl}h_{ml}=0,\quad \textrm{for }j,m=0,\dots,d-1, \, j\neq m.
\end{equation}
Two Hadamard matrices $H$ and $H'$ are called \emph{equivalent} if there exist permutation matrices
$P$, $P'$ and diagonal unitary matrices $D$, $D'$ such that $H'=DPHP'D'$.

Jedwab and Wiebe \cite{Wie13,JedWie15,JedWie16} have recently proposed a simple method of constructing SIC-POVMs
in certain dimensions, which employs complex Hadamard matrices. We recall it briefly below.

Let $H$ be a complex Hadamard $d\times d$ matrix and let $v \in \mathbb{C}$. Consider the set $H(v):=\{ H_{jk}(v)\}_{j,k=0}^{d-1}$ of $d^2$ vectors
in $\mathbb C^d$ such that $H_{jk}(v)$ is the $j$-th row of $H$ with the $k$-th coordinate multiplied by $v$. Denoting the canonical orthonormal
basis in $\mathbb C^d$ by $(e_l)_{l=0}^{d-1}$ we can write $H_{jk}(v)$ as $H_{jk}(v)=\sum_{l=0}^{d-1} h_{jl}e_l+(v-1) h_{jk}e_k$.
Jedwab and Wiebe proved in \cite[Theorem 1]{JedWie15} that $H(v)$ generates a set of $d^2$ equiangular lines in $\mathbb C^d$ if and only if:

\begin{itemize}
\item $d=2$ and $v \in \{ \pm(1\pm\sqrt 3)(1\pm i)/2\}$, or
\item $d=3$ and $v\in\{0,-2,1\pm\sqrt3 i\}$, or
\item $d=8$, $H$ is equivalent to a (unique up to equivalence) real Hadamard matrix and $v\in\{-1\pm 2i\}$.
\end{itemize}
Moreover, for $d=8$ the obtained sets of equiangular lines are the Hoggar lines.
On the other hand, for $d=2$ every complex Hadamard matrix is necessarily equivalent to a real Hadamard matrix, 
and all the SIC-POVMs are isomorphic to the `tetrahedral' one.

From now on, we shall denote the SIC-POVM
corresponding to $H(v)$ by the same letter if no confusion arises.
Now we can formulate the main results of the present paper:

\begin{Theorem}
\label{T1}
Let a complex Hadamard matrix $H$ in dimension $d\in\{2,8\}$
and $v \in \mathbb{C}$ be such that $H(v):=\{ H_{jk}(v)\}_{j,k=0}^{d-1}$
forms a set of equiangular vectors. Then the entropy of $H(v)$ is minimized by $d^2$ states in $H(\bar{v})$. Moreover, the minimal value of entropy
is $\ln(d(d+1)/2)$.
\end{Theorem}

\begin{proof}
Set $m,n=0,\ldots,d-1$.
First, we show that the sequence
$T_{mn}:=(|H_{jk}(v)\cdot H_{mn}(\bar{v})|^2)_{j,k=0}^{d-1}$
consists of only two elements, one of which is $0$.
We know that there exist a real Hadamard matrix $H'$ and diagonal unitary matrices
$D=diag(c_1,\ldots,c_d)$ and $D'=diag(c'_1,\ldots,c'_d)$ such that $H=DH'D'$.
Clearly, $(e'_l)_{l=0}^{d-1}$, where $e'_l := c'_l \cdot e_l$ $(l=0,\dots,d-1)$, is
also an orthonormal basis of $\mathbb{C}^d$.
Then
$$H_{jk}(v) = c_j \left( \sum_{l=0}^{d-1} h'_{jl}e'_l+(v-1)h'_{jk}e'_k \right)$$
for $j,k=0,\ldots,d-1$, and so
$$|H_{jk}(v)\cdot H_{mn}(\bar{v})|=|H'_{jk}(v)\cdot H'_{mn}(\bar{v})|$$
for $j,k,m,n=0,\ldots,d-1$. This identity reduces calculations to the real case,
and so from now on we assume that $H$ is a real Hadamard matrix.

Now, using (\ref{diag}) and (\ref{offdiag}), we get
\begin{multline*}
|H_{jk}(v)\cdot H_{mn}(\bar{v})|^2 = |\sum_{l,r=0}^{d-1} h_{jl}h_{mr}e_l\cdot e_r \\
+\sum_{l=0}^{d-1} (v-1)(h_{jl}h_{mn}e_l\cdot e_n+h_{jk}h_{ml}e_k\cdot e_l) \\
+(v-1)^2h_{jk}h_{mn}e_k\cdot e_n|^2 = \\
|d\delta_{jm}+(v-1)(h_{jn}h_{mn}+h_{jk}h_{mk})+(v-1)^2h_{jk}h_{mn}\delta_{kn}|^2.
\end{multline*}

In particular, for $j\neq m$ and $k\neq n$ we have
\begin{equation}
\label{h}
|H_{jk}(v)\cdot H_{mn}(\bar{v})|^2=|(v-1)(h_{jn}h_{mn}+h_{jk}h_{mk})|^2.
\end{equation}

It follows from (\ref{offdiag}) and from the fact that the entries of $H$ are $\pm 1$
that for all $m,n,j=0,\dots,d-1$ and $j\neq m$ there exist exactly $d/2$
such $k=0,\dots,d-1$ that the above expression is equal to $0$.
Otherwise, it is $|2v-2|^2$.

For $j\neq m$ and $k=n$ we get
$$|H_{jk}(v)\cdot H_{mk}(\bar{v})|^2=|v^2-1|^2;$$
on the other hand, for $j=m$ and $k\neq n$ we obtain
$$|H_{jk}(v)\cdot H_{jn}(\bar{v})|^2=|d+2v-2|^2,$$
and finally, for $j=m$ and $k=n$ we have
$$|H_{jk}(v)\cdot H_{jk}(\bar{v})|^2=|d+v^2-1|^2.$$

Now, straightforward calculations show that for the values of $v$ obtained by
Jedwab and Wiebe \cite{JedWie15} all the $d(d+1)/2$ non-zero members of the sequence $T_{mn}$
are equal and depend only on $d$ and $v$. This, in turn, implies that $T_{mn}$ attains value
$0$ with multiplicity $(d-1)d/2$.

Let us now consider the pre-measurement state generated by $H_{mn}(\bar{v})$
for $m,n=0,\ldots,d-1$, and the SIC-POVM $H(v)$. In this case, as has just been shown, the distribution of measurements outcomes provides us with $d(d+1)/2$
probabilities equal to $2/(d(d+1))$ and $(d-1)d/2$ equal to $0$. According to
the result discussed in Sec.\thinspace \ref{Infpow}, the state generated by $H_{mn}(\bar{v})$ must be a minimizer for the entropy
of $H(v)$ and the minimal value is equal to $\ln(d(d+1)/2)$.
\end{proof}

\begin{Theorem}
Under the assumptions of Theorem \ref{T1}, the informational power of
$H(v)$ is equal to $\ln(2d/(d+1))$,
and the states generated by the vectors in $H(\bar{v})$ constitute an
equiprobable maximally informative ensemble. Particularly, the informational
power of Hoggar lines is $2\ln(4/3)$.
\end{Theorem}

\begin{proof}
Since $\operatorname*{Sym}(H(v))$ acts irreducibly on
$\mathcal{P}(\mathbb{C}^d)$, the equality in (\ref{infpowent}) holds
for $\Pi = H(v)$. Hence, applying (\ref{infpowent}) and Theorem \ref{T1}, we get
$$W(\Pi) = \ln(d^2) - \ln(d(d+1)/2)= \ln(2d/(d+1)).$$
Then $d^2$ equiprobable states corresponding to the vectors from $H(\bar{v})$ form a maximally informative ensemble.
\end{proof}

\section{Much ado about zeros}

\vspace{-0.1cm}

In the above reasoning the zeros of the probability distribution of measurement
outcomes play a key role. We already know that for the pre-measurement state of
the system being the entropy minimizer, their number is maximal and equals $28$
for Hoggar lines, see also \cite{AppEriFuc11}. Let us now have a closer look at
the localization of these $28$ zeros for $64$ minimizers described by Theorem \ref{T1}.

From now on we label the elements of the Hoggar lines SIC-POVM $H(v)$
by the elements of $\Sigma := \mathbb{Z}_{2}^{3} \otimes \mathbb{Z}_{2}^{3}$,
the translation group of the six-dimensional affine space over $GF(2)$,
isomorphic to $P_3$ acting regularly on $H(v)$.
Moreover, we assume for definiteness that $H$ is the (real) Sylvester-Hadamard matrix
$H_3$ considered, e.g.\ in \cite{JedWie15}, writing the indices in the binary expansion
as elements of $\mathbb{Z}_{2}^{3}$. In this case we have $h_{\iota\kappa}=(-1)^{\iota_1\kappa_1+\iota_2\kappa_2+\iota_3\kappa_3}$ 
for $\iota,\kappa\in \mathbb{Z}_{2}^{3}$. Moreover, the standard representation of the three-qubit Pauli group, constructed
from the Pauli matrices $\sigma_X$ and $\sigma_Z$, acts (up to a phase) on vectors in $H(v)$ and $H(\bar v)$ in the following way:
\begin{equation*}
(\sigma_Z^{\alpha_1}\sigma_X^{\beta_1}\otimes\sigma_Z^{\alpha_2}\sigma_X^{\beta_2}\otimes\sigma_Z^{\alpha_3}\sigma_X^{\beta_3})
H_{\iota\kappa}(w) = H_{\iota+\alpha,\kappa+\beta}(w)
\end{equation*}
for $\iota,\kappa,\alpha,\beta\in \mathbb{Z}_{2}^{3},w=v,\bar v$. Consider now the blocks
\begin{equation*}
B_{\mu\nu} := \{(\iota,\kappa) : H_{\iota\kappa}(v)\cdot H_{\mu\nu}(\bar{v}) = 0, \iota,\kappa \in \mathbb{Z}_{2}^{3}\},
\end{equation*}
of zeros of $T_{\mu\nu}$ for $\mu,\nu \in \mathbb{Z}_{2}^{3}$, where $T_{\mu\nu}$ is as in the proof of Theorem \ref{T1}. It follows from (\ref{h}) that
$(\iota,\kappa) \in B_{\mu\nu}$ if and only if $\iota \neq \mu$, $\kappa \neq \nu$, and $h_{\mu\nu}h_{\iota\nu}+h_{\mu\kappa}h_{\iota\kappa}=0$,
or equivalently $h_{\mu+\iota,\nu+\kappa}=-1$, for $\iota,\kappa\in \mathbb{Z}_{2}^{3}$. Hence
$B_{00} = \{(\iota,\kappa) :h_{\iota\kappa}=-1, \iota,\kappa \in \mathbb{Z}_{2}^{3}\}$
and $B_{\mu\nu} = B_{00} + (\mu,\nu)$ for $\mu,\nu \in \mathbb{Z}_{2}^{3}$.
It is easy to show that
$\mathscr{B}_8 := \{B_{\mu\nu}\}_{\mu,\nu \in \mathbb{Z}_{2}^{3}} \subset \Sigma$ constitutes a \emph{symmetric} (\emph{Menon}) $(64,28,12)$-\emph{design},
see \cite{MooPol13} for terminology from design theory.
Moreover, this design is the development of the respective \emph{difference set} in $\Sigma$. More precisely, one can show that $\mathscr{B}_8$ is so called \emph{sympletcic design}
${\mathscr{S}}^{-1}(6)$ analysed by Kantor in \cite{Kan75}. He proved that $\operatorname*{Aut}({\mathscr{S}}^{-1}(6))$, the automorphism group
of ${\mathscr{S}}^{-1}(6)$, is a semidirect product
of $\Sigma$ by the symplectic group $Sp(6,2)$, i.e.\ the group of linear transformations of the vector space $\mathbb{Z}_2^{6} \simeq \Sigma$ over $GF(2)$ preserving the natural symplectic form. More precisely, for all $(\iota,\kappa) \in \Sigma$ and $M \in Sp(6,2)$ the respective affine transformation sends $B_{\mu\nu}$ onto $B_{M(\mu,\nu)+(\iota,\kappa)}$ for all $\mu,\nu \in \mathbb{Z}_{2}^{3}$.
Moreover, $\operatorname*{Aut}({\mathscr{S}}^{-1}(6))$ acts $2$-transitively on blocks \cite{Kan75}.

\vspace{-0.3cm}

\section{Twin sets of Hoggar lines}

\vspace{-0.1cm}

Now, let us have a closer look at the set $H(\bar{v})$, all of whose elements
are minimizers for the entropy of $H(v)$,
and form a maximally informative ensemble for this measurement.
It follows from Theorem~\ref{T1} and \cite[Theorem 1]{JedWie15} that
$H(\bar{v})$ is also: the `tetrahedral' POVM for $d=2$, and the set
of Hoggar lines for $d=8$. The question arises, how these two subsets
of $\mathcal P(\mathbb {C}^d)$, $H(v)$ and $H(\bar{v})$, are related
to one another. Let $C: \mathbb {C}^d \rightarrow \mathbb {C}^d$ be
a (\emph{complex}) \emph{conjugation with respect to the basis} $(e'_l)_{l=0}^{d-1}$ from the proof of Theorem 1,
i.e.\ an antiunitary involutive map keeping the basis invariant \cite{Garetal14}, given by
$C(\sum_{l=0}^{d-1} x_l e'_l) := \sum_{l=0}^{d-1} \bar{x_l} e'_l$ for $(x_l)_{l=0}^{d-1} \in \mathbb {C}^d$.
Then $H(\bar{v})$ is the image of $H(v)$ under the collineation generated by $C$; more precisely, $H_{jk}(\bar{v}) = C (H_{jk}(v))$ for
$j,k=0,\dots,d-1$.

To express the relationship between $H(v)$ and $H(\bar{v})$ more
geometrically, we can use the generalized Bloch representation.
For $d=2$, these SIC-POVMs are represented on the Bloch sphere as
two dual regular tetrahedra that together form a stellated octahedron a.k.a.
\emph{stella octangula}. For $d=8$ we get in the generalized Bloch
representation two regular $63$-dimensional simplices inscribed in the
unit sphere in a $63$-dimensional real vector space, where one is the image of
the other under a reflection through a $35$-dimensional linear subspace.
It is so, because in the generalized Bloch representation of quantum states as
elements of the unit sphere of the real $(d^2-1)$-dimensional vector space of traceless
Hermitian $d \times d$ matrices, a conjugation map acting on $\mathbb {C}^d $ is transformed
into a transpose operation (both defined in the same basis), see e.g.\ \cite[p.~4]{Leietal06}.
Under this operation only traceless symmetric real matrices are invariant, and they form
a $(d+2)(d-1)/2$-dimensional vector subspace.

Moreover, it turns out that for $d=8$ the sets $H(v)$ and $H(\bar v)$ correspond to `twin' sets of Hoggar lines
considered in \cite[Sec.~2.3]{Zhu15a}. Let $\psi$ be a fiducial vector for some $HL$. Zhu showed that
there is an order-7 unitary $U_7$ in $\left(\operatorname*{Sym}\left(HL\right)\right)_\psi$ with six one-
and one two-dimensional eigenspaces, such that the latter contains both $\psi$ and its `twin' vector, ${\psi}'$,
which also generates (another) set of Hoggar lines $HL'$, lying on the same orbit under action of the Clifford group.
To be more specific, assume again that $H=H_3$. Let $\psi$ and $\psi'$ be given, respectively, by Eqs.~(14) and (3)
in \cite{Zhu15a}. Then, all four sets of Hoggar lines: $H(v)$, $H(\bar v)$, and those generated by $\psi$ and $\psi'$,
are covariant with respect to the standard representation of the three-qubit Pauli group.
Let $U$ denote a Clifford unitary for this group from \cite[p.~2]{JedWie15}. Now, observe that, up to a normalization factor,
$U\psi = H_{(0,0,0)(0,1,1)}(v)$ and $U\psi'=H_{(1,0,1)(0,0,0)}(\bar{v})$,
and they are indeed fiducial vectors, respectively, for $H(v)$ and $H(\bar v)$, lying in the same two-dimensional
eigenspace of an order-7 unitary $UU_7U^{\dagger}$.

Finally, note that the symmetry groups of both Hoggar SIC-POVMs, $H(v)$ and $H(\bar{v})$, are identical.
It follows from the fact that the symmetry groups of the `twin' sets of Hoggar lines $HL$ and $HL'$ described above are the same.
Indeed, these symmetry groups are generated by the same representation of the three-qubit Pauli group and, respectively, the stabilizers
of $\psi$ and ${\psi}'$.
Thus, it suffices to show that the stabilizer of ${\psi}'$ is contained in the symmetry group of $HL$.
The stabilizer has two generators: $U_7$, which stabilizes both fiducials, and $U_{12}$, an order-12 unitary defined
in \cite[Sec.10.4]{Zhu12}. By straightforward calculation, we get that $U_{12}$ permutes the elements of $HL$
and so belongs to its symmetry group. The situation is similar for two dual `tetrahedral' POVMs in
$d=2$ sharing also the same symmetry group.

\bibliography{PRA}

\end{document}